\PassOptionsToPackage{noamssymbols}{newtxmath}
\documentclass[sigconf,nonacm]{acmart}

\usepackage{amsmath,amssymb,mathtools}
\usepackage{booktabs}
\usepackage{tikz}
\usetikzlibrary{arrows.meta,positioning,calc,fit,backgrounds,shapes.geometric}
\usepackage{pgfplots}
\pgfplotsset{compat=1.16}
\usepackage{float}
\newfloat{algorithm}{tbp}{loa}
\floatname{algorithm}{Algorithm}
\usepackage{algpseudocode}

\usepackage[T1]{fontenc}
\usepackage[utf8]{inputenc}

\setcopyright{none}
\renewcommand\footnotetextcopyrightpermission[1]{}
\newcommand{\R}{\mathbb{R}}
\newcommand{\N}{\mathbb{N}}
\newcommand{\E}{\mathbb{E}}
\let\Pr\relax\newcommand{\Pr}{\mathbb{P}}
\newcommand{\A}{\mathcal{A}}
\newcommand{\X}{\mathcal{X}}

\newcommand{\D}{\mathcal{D}}

\newcommand{\Alg}{\mathfrak{A}}
\newcommand{\rel}{\rho}
\newcommand{\cost}{c}
\newcommand{\eqdef}{\triangleq}
\newcommand{\odds}{o}
\newcommand{\LR}{\Lambda}
\newcommand{\seq}{\mathbin{\fatsemi}}
\newcommand{\fatsemi}{\mathbin{;}}
\DeclareMathOperator*{\maj}{maj}
\DeclareMathOperator*{\argmax}{arg\,max}
\newcommand{\DKL}{D_{\mathrm{KL}}}
\begin{document}

\title{Odds Law: The Decomposition Algebra}
\subtitle{On How Intelligence Organizes Itself to Solve Difficult Problems
Reliably}

\author{Hidayet Aksu}
\affiliation{%
  \institution{}
  \city{}
  \country{}
}
\email{hidayetaksu@gmail.com}

\renewcommand{\shortauthors}{}

\begin{abstract}
We ask a structural question: \emph{given unreliable elementary
problem-solvers, what organizations of them solve hard problems
reliably, and what are the limits?} We develop a \emph{decomposition
algebra}: elementary solvers are morphisms in a Markov (stochastic)
category, and four combinators (sequential composition, parallel
ensembling, verification gating, and recursive reduction) generate the
space of compound solvers. We equip this algebra with two homomorphisms,
a \emph{reliability} valuation into the ordered monoid $([0,1],\le)$ and
a \emph{cost} valuation into a commutative semiring, and we derive the
composition laws that govern how reliability flows through structure. Our
central results are (i) a \emph{verification odds law} (the result that
names this report), showing that a
verification gate multiplies the odds of correctness by the verifier's
likelihood ratio $\LR$, so that $k$ conditionally independent gates yield
geometric amplification; (ii) a \emph{reliability amplification theorem},
giving target reliability $1-\delta$ at $O(\log 1/\delta)$ verification
depth whenever $\LR>1$; and (iii) a \emph{threshold dichotomy}: above the
critical parameters ($\LR^\star=1$ for verification, $p^\star=\tfrac12$
for majority voting) reliability can be driven arbitrarily close to one
at logarithmic cost, while at or below them no amplification is possible.
We then show that \emph{self-organization} (the spontaneous appearance
of layered, verifier-saturated structure) is the least fixed point of a
monotone improvement operator on the complete lattice of strategies, and
that this fixed point equalizes marginal log-odds gain per unit cost
(a water-filling characterization). Finally, we prove matching limits:
an information ceiling bounds per-gate amplification by a divergence
quantity; shared error causes create a strictly positive voting
floor (we characterize exactly when), so diversity is \emph{necessary}
for unbounded amplification; and a no-free-lunch corollary shows that,
averaged over all problem families, no decomposition beats its base
solver. Reliability, in short, is neither free nor magical: it is bought
with independent information, arranged by composition, and bounded by the
verifier. These results form \textsc{Odds Law}, the theory layer of a
two-part program (Theory: \textsc{Odds Law} $\rightarrow$ Framework:
\textsc{Maestro Order}); the companion report builds the
\textsc{Maestro Order} harness on these laws.
\end{abstract}


\maketitle

\section{Introduction}
A single attempt at a hard problem is rarely trustworthy. Yet
collections of unreliable attempts, suitably \emph{organized}, can be
made arbitrarily trustworthy; this is the everyday experience of
science, engineering, bureaucracy, and biological computation. Von
Neumann's classical question, how to build a reliable organism (or
automaton) from unreliable components~\cite{vonneumann1956}, is the same
question one asks of a research group, a compiler with its test suite, or
an ensemble of fallible reasoners. The components differ; the
\emph{organizing principle} may not.

This paper takes that organizing principle as its object of study. We fix
a population of elementary problem-solvers (each a randomized map from
problem instances to candidate answers, and each only partly reliable) and
ask which \emph{compositions} of them are reliable, at what cost, and
against what fundamental limits. We deliberately abstract away the
internals of a solver (whether it is a person, a heuristic, a learned
model, or an exact algorithm) and study only the calculus by which
solvers are combined. The thesis is that reliable problem-solving is a
property of \emph{structure}, and that this structure obeys algebraic
laws.

\paragraph{The four combinators.}
Empirically, intelligent systems combine solvers in a small number of
recurring ways. They \emph{decompose} a problem into ordered subproblems
and solve them in sequence (planning, proofs, pipelines). They run
\emph{several} attempts and aggregate (committees, juries,
self-consistency, ensembles). They \emph{check} candidate answers and
keep only those that pass (proofs verified, code unit-tested, claims
corroborated). And they \emph{recurse}, treating a subproblem with the
very same repertoire. We argue that these four (sequential composition,
parallel ensembling, verification gating, and recursion) are not an
arbitrary list but a generating set for a well-behaved algebra of
solvers, in which reliability and cost are homomorphic images of
structure.

\paragraph{Contributions.}
\begin{itemize}
\item \textbf{A decomposition algebra} (\S\ref{sec:model}--\S\ref{sec:algebra}).
We model solvers as morphisms in the Kleisli category of the
subdistribution monad (a Markov category~\cite{fritz2020,chojacobs2019}),
and define four combinators that generate the free algebra $\Alg(B)$ over
a base set $B$. Reliability and cost are valuations into an ordered
monoid and a commutative semiring, respectively.

\item \textbf{Composition laws} (\S\ref{sec:laws}). We prove how each
combinator transforms reliability: sequential composition degrades
(Lemma~\ref{lem:seq}), majority voting amplifies above a threshold
(Lemma~\ref{lem:vote}), and, as our key primitive, a verification gate
multiplies the \emph{odds} of correctness by the verifier's likelihood
ratio (Lemma~\ref{lem:odds}).

\item \textbf{Amplification and a threshold dichotomy}
(\S\ref{sec:amp}). Theorem~\ref{thm:amp} achieves reliability $1-\delta$
at verification depth $O(\log 1/\delta)$ whenever the verifier is
informative; Theorem~\ref{thm:threshold} shows this is sharp: at the
critical parameters, no amplification is possible. This is the
problem-solving analogue of fault-tolerance threshold
theorems~\cite{vonneumann1956,aharonov1997}.

\item \textbf{Self-organization as a fixed point} (\S\ref{sec:lattice}).
Strategies form a complete lattice under refinement; a greedy,
budget-aware improvement operator is monotone, so by Knaster--Tarski
\cite{tarski1955} it has a least fixed point, the canonical
\emph{self-organized} strategy, which we characterize as
verifier-saturated and marginal-rate-equalizing.

\item \textbf{Matching limits} (\S\ref{sec:lower}). An information ceiling
bounds per-gate amplification (Theorem~\ref{thm:info}); shared error
causes impose a strictly positive voting floor exactly when they can push
the committee below chance, so diversity is \emph{necessary}
(Theorem~\ref{thm:diversity}); and a no-free-lunch corollary
\cite{wolpert1997} shows decomposition gains are paid for with priors
matched to the problem.
\end{itemize}

A companion report instantiates this algebra as a concrete,
model-agnostic orchestration harness (\textsc{Maestro Order}) and
measures the predicted laws; here we develop the theory in the abstract.

\section{The Solver Model}\label{sec:model}
We work over a \emph{problem family}: a measurable space of instances
$\X$, an answer space $\A$ (with a distinguished symbol $\bot\notin\A$
for \emph{abstention}), and a \emph{correctness oracle}
$Y^\star:\X\to 2^{\A}\setminus\{\emptyset\}$ assigning to each instance
its nonempty set of acceptable answers. A distribution $\D$ over $\X$
specifies which instances occur and how often.

\begin{definition}[Solver]\label{def:solver}
A \emph{solver} is a Markov kernel $s:\X\rightsquigarrow \A\cup\{\bot\}$,
i.e.\ a measurable map $\X\to\Delta(\A\cup\{\bot\})$ from instances to
(sub)distributions over answers and abstention. Write
$s(\cdot\mid x)$ for the output law on instance $x$.
\end{definition}

Two scalar valuations summarize a solver. Let $\mathrm{cov}(s)\eqdef
\Pr_{x\sim\D}[\,s(x)\neq\bot\,]$ be its \emph{coverage} (the probability
it commits to an answer).

\begin{definition}[Reliability]\label{def:rel}
The \emph{reliability} of $s$ is the conditional correctness rate on
committed answers,
\[
\rel(s)\;\eqdef\;
\Pr_{x\sim\D,\;a\sim s(x)}\!\big[a\in Y^\star(x)\,\big|\,a\neq\bot\big],
\]
with the convention $\rel(s)=0$ if $\mathrm{cov}(s)=0$. We also use the
\emph{worst-case} variant
$\underline{\rel}(s)=\inf_{x\in\X}\Pr_{a\sim s(x)}[a\in
Y^\star(x)\mid a\neq\bot]$ when distribution-freeness is required.
\end{definition}

Separating coverage from reliability is deliberate: a solver may raise
reliability by abstaining more (refusing hard instances), and the
coverage/reliability trade-off is a recurring theme
(\S\ref{sec:laws},~\S\ref{sec:frontier}).

\begin{definition}[Cost]\label{def:cost}
A \emph{cost} valuation assigns to $s$ a value $\cost(s)$ in a commutative
semiring $(\mathcal{K},\oplus,\otimes,0,1)$. We instantiate
$\mathcal{K}=(\R_{\ge0}\cup\{\infty\},+,\cdot)$, with $\cost(s)$ the
expected number of \emph{base-solver invocations} used by $s$. Sequential
work adds ($+$); independent repetition multiplies counts by branching
factor.
\end{definition}

\begin{table}[t]
\centering\small
\caption{Notation.}
\label{tab:notation}
\begin{tabular}{@{}ll@{}}
\toprule
Symbol & Meaning \\
\midrule
$\X,\A,\bot$ & instances, answers, abstention \\
$Y^\star(x)$ & acceptable-answer set (oracle) \\
$s,g,v$ & solver, generator, verifier \\
$\rel(s)$ & reliability (Def.~\ref{def:rel}) \\
$\mathrm{cov}(s)$ & coverage (commit probability) \\
$\cost(s)$ & expected base invocations \\
$\beta,\alpha$ & verifier completeness, false-acceptance \\
$\LR=\beta/\alpha$ & verifier discrimination (likelihood ratio) \\
$\odds=\rel/(1{-}\rel)$ & odds; $\ell=\log\odds$ log-odds \\
$\seq,\oplus_A,V_v,\mu$ & sequential, vote, verify, recurse \\
$\Alg(B)$ & decomposition algebra over base $B$ \\
\bottomrule
\end{tabular}
\end{table}

\paragraph{Categorical view.}
Solvers compose. Fix the Kleisli category $\mathbf{Stoch}$ of the
subdistribution monad: objects are (typed) problem spaces, a morphism
$A\rightsquigarrow B$ is a Markov kernel, and composition is the
Chapman--Kolmogorov integral. $\mathbf{Stoch}$ is a \emph{Markov
category}~\cite{fritz2020,chojacobs2019}: it is symmetric monoidal under
the product $\otimes$, with copy/discard structure modelling
duplication and erasure of intermediate results. Sequential
decomposition is categorical composition; running solvers ``side by
side'' is the monoidal tensor. This is the ambient category in which our
algebra lives; we keep the categorical language light and verify all
quantitative claims directly. Readers unfamiliar with category theory can
read this paragraph as saying only that randomized solvers can be wired
together in sequence and in parallel, and that both operations behave
well.

\section{The Decomposition Algebra}\label{sec:algebra}
Let $B$ be a finite set of \emph{base solvers}. The decomposition algebra
$\Alg(B)$ is the smallest set of solvers containing $B$ and closed under
the four combinators below. Figure~\ref{fig:combinators} depicts them.

\begin{figure}[t]
\centering
\begin{tikzpicture}[
  box/.style={draw,rounded corners,minimum height=6mm,minimum width=9mm,font=\scriptsize},
  >=Latex,node distance=5mm,font=\scriptsize]
\node[box] (s1) {$s_1$};
\node[box,right=of s1] (s2) {$s_2$};
\node[box,right=of s2] (s3) {$s_3$};
\draw[->] (s1)--(s2); \draw[->] (s2)--(s3);
\node[left=1mm of s1] {$x$}; \node[right=1mm of s3] {$y$};
\node[above=0.5mm of s2,font=\scriptsize\bfseries] {(a) sequential $\seq$};
\begin{scope}[yshift=-23mm]
\node[box] (p2) {$s$};
\node[box,above=2mm of p2] (p1) {$s$};
\node[box,below=2mm of p2] (p3) {$s$};
\node[draw,circle,right=8mm of p2,inner sep=1pt] (agg) {$A$};
\draw[->] (p1)-|(agg); \draw[->] (p2)--(agg); \draw[->] (p3)-|(agg);
\node[left=1mm of p2] {$x$}; \node[right=1mm of agg] {$y$};
\node[above=2mm of p1,font=\scriptsize\bfseries] {(b) parallel $\oplus_A$};
\end{scope}
\begin{scope}[yshift=-50mm]
\node[box] (g) {$g$};
\node[draw,diamond,right=7mm of g,inner sep=0pt,aspect=1.6] (v) {$v$};
\draw[->] (g)--(v);
\draw[->] (v) to[out=120,in=60,looseness=2] node[above,font=\tiny]{reject} (g);
\node[left=1mm of g] {$x$}; \node[right=1mm of v] {accept $\to y$};
\node[font=\scriptsize\bfseries] at ($(g)!0.5!(v)+(0,11mm)$) {(c) verify--refine $V_v$};
\end{scope}
\begin{scope}[yshift=-68mm]
\node[box] (r) {$r$};
\node[box,right=6mm of r] (sub) {$\Alg$};
\node[box,right=6mm of sub] (m) {$\textstyle\bigsqcup$};
\draw[->] (r)--(sub); \draw[->] (sub)--(m);
\draw[->] (sub) to[out=300,in=240,looseness=1.5] (sub);
\node[left=1mm of r] {$x$}; \node[right=1mm of m] {$y$};
\node[above=2mm of sub,font=\scriptsize\bfseries] {(d) recursion $\mu$};
\end{scope}
\end{tikzpicture}
\caption{The four generators of the decomposition algebra. (a) ordered
subproblems; (b) replicated solvers aggregated by $A$; (c) a generator
$g$ gated by a verifier $v$ with refinement on rejection; (d) a reducer
$r$ that emits subproblems solved recursively and recombined.}
\label{fig:combinators}
\end{figure}
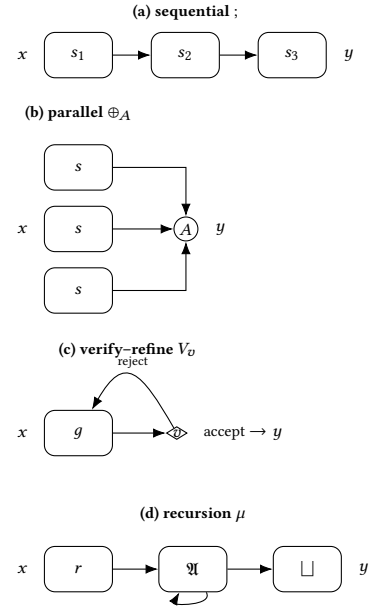

\paragraph{(1) Sequential composition $\seq$.}
For solvers $s_1,\dots,s_k$ whose types chain
($s_{i+1}$ consumes the output of $s_i$), the pipeline
$s_k\seq\cdots\seq s_1$ is their Kleisli composite. It models
decomposition into \emph{dependent} subproblems: the final answer is
correct only if each stage produces a usable intermediate (we make the
dependence precise in Lemma~\ref{lem:seq}).

\paragraph{(2) Parallel ensembling $\oplus_A$.}
Given a solver $s$, a replication count $n$, and an \emph{aggregator}
$A:(\A\cup\{\bot\})^n\rightsquigarrow \A\cup\{\bot\}$, the ensemble
$\oplus_A(s,n)$ draws $n$ independent samples $a_1,\dots,a_n\sim s(x)$ and
returns $A(a_1,\dots,a_n)$. The canonical aggregator is plurality vote
$A=\maj$.

\paragraph{(3) Verification gating $V_v$.}
A \emph{verifier} is a kernel $v:\X\times\A\rightsquigarrow
\{\textsf{acc},\textsf{rej}\}$. Given a generator $g$ and budget
$T\in\N$, the gate $V_v(g,T)$ repeatedly samples $a\sim g(x)$, returns
the first $a$ with $v(x,a)=\textsf{acc}$, and abstains ($\bot$) after $T$
rejections. A verifier is summarized by its \emph{completeness}
$\beta\eqdef\Pr[v=\textsf{acc}\mid a\in Y^\star(x)]$ and \emph{soundness}
$1-\alpha$, where $\alpha\eqdef\Pr[v=\textsf{acc}\mid a\notin
Y^\star(x)]$ is its false-acceptance rate. Its \emph{discrimination} is
the likelihood ratio $\LR\eqdef\beta/\alpha\in[0,\infty]$.

\paragraph{(4) Recursion $\mu$.}
A \emph{reducer} $r:\X\rightsquigarrow \X^{\le m}$ maps an instance to a
finite tuple of subinstances, with a recombiner
$\bigsqcup:\A^{\le m}\rightsquigarrow\A$. The recursive solver
$\mu\,F$ is the least solution of $S = F(S)$ where
$F(S)= \bigsqcup\circ\, S^{\otimes}\circ r$ applies $S$ to each
subinstance. Well-posedness (a least fixed point exists) is established
in \S\ref{sec:lattice}.

\begin{definition}[Decomposition algebra]\label{def:algebra}
$\Alg(B)$ is the closure of $B$ under $\{\seq,\oplus_A,V_v,\mu\}$. An
element of $\Alg(B)$ is a \emph{strategy}; its syntax tree is its
\emph{organization}.
\end{definition}

\paragraph{The algebra is not a metaphor.}
The four combinators are the formal skeleton of organizations that already
work. Table~\ref{tab:instances} lists canonical instantiations across very
different substrates; in each case the ``verifier'' is whatever cheaply
certifies a candidate, and the reliability laws of \S\ref{sec:laws} apply
unchanged. That a jury, a compiler's test suite, a proof checker, and a
MapReduce job are all the same algebra evaluated on different base solvers
is the unifying claim of this paper.

\begin{table}[t]
\centering\footnotesize
\caption{Canonical instantiations of the decomposition algebra. Each row
fixes a base solver and a verifier; the combinators and laws are shared.}
\label{tab:instances}
\begin{tabular}{@{}llll@{}}
\toprule
System & Base solver & Verifier $v$ &
\begin{tabular}[t]{@{}l@{}}Dominant\\combinator\end{tabular} \\
\midrule
Math proof & prover & proof checker & $V_v$ (high $\LR$) \\
Software & coder & test suite / types & $V_v\!\seq$ \\
Jury / panel & juror & --- & $\oplus_{\maj}$ \\
Ensemble ML & weak learner & --- & $\oplus_{\maj}$ \\
Science & lab/group & replication & $V_v\,\oplus$ \\
MapReduce & mapper & re-execution & $\mu\,\seq$ \\
\bottomrule
\end{tabular}
\end{table}

\paragraph{Valuations are (lax) homomorphisms.}
Cost is an exact homomorphism into $\mathcal{K}$:
$\cost(s_k\seq\cdots\seq s_1)=\sum_i\cost(s_i)$,
$\cost(\oplus_A(s,n))=n\,\cost(s)+\cost(A)$, and
$\cost(V_v(g,T))=\E[N]\,(\cost(g)+\cost(v))$ with $N\le T$ the number of
rounds. Reliability is a \emph{lax} homomorphism: it does not factor
through structure exactly, but is bounded above and below by explicit
functions of the children's reliabilities. Deriving those bounds is the
goal of the next section.

\paragraph{Two bands: the log-odds semiring.}
Reliability composes most cleanly not in $[0,1]$ but in \emph{log-odds}.
Write $\ell(s)\eqdef\log\frac{\rel(s)}{1-\rel(s)}\in\R\cup\{\pm\infty\}$.
The point of the next section is that the natural operations live in two
algebraic bands: a \emph{reliability band}, where verification cascades
\emph{add} log-odds, and a \emph{cost band}, where work \emph{adds}
invocation counts.

\begin{proposition}[Log-odds additivity]\label{prop:semiring}
On the sub-algebra generated by verification gates, the map
$\ell$ is a monoid homomorphism from conditionally independent gate
cascades (under composition) to $(\R,+)$: stacking gates $V_{v_1},\dots,V_{v_k}$ sends
$\ell\mapsto \ell+\sum_i\log\LR_i$. Pairing $\ell$ with cost
$\cost\in(\R_{\ge0},+)$ yields a graded monoid $(\R\times\R_{\ge0},+)$ in
which the slope $\Delta\ell/\Delta\cost$ is the marginal rate optimized in
\S\ref{sec:lattice}.
\end{proposition}
\begin{proof}
Immediate from Lemma~\ref{lem:odds} and Theorem~\ref{thm:amp}:
$\ell$ after a cascade is $\ell_0+\sum_i\log\LR_i$, additive and
associative with identity the uninformative gate ($\log\LR=0$). Cost adds
by Definition~\ref{def:cost}. The product monoid is the stated grading.
\end{proof}

This is why log-odds is the right currency: it linearizes the strongest
combinator and turns ``how to organize'' into a linear-programming
intuition over marginal slopes.

\section{Reliability Composition Laws}\label{sec:laws}
Throughout, treat distinct invocations of a base solver as independent
unless stated otherwise; \S\ref{sec:lower} removes this assumption and
shows how much depends on it.

\subsection{Sequential composition degrades}
\begin{lemma}[Serial law]\label{lem:seq}
Let $s_1,\dots,s_k$ be stages whose composite is correct iff every stage
is correct (no error masking), with per-stage reliabilities $\rel_i$.
Then
\[
1-\sum_{i=1}^{k}(1-\rel_i)\;\le\;\rel(s_k\seq\cdots\seq s_1)\;\le\;
\min_i \rel_i,
\]
and if stage errors are independent,
$\rel(s_k\seq\cdots\seq s_1)=\prod_{i=1}^k \rel_i$.
\end{lemma}
\begin{proof}
Let $E_i$ be the event that stage $i$ errs. Correctness is
$\bigcap_i \overline{E_i}$. The upper bound is monotonicity:
$\Pr[\bigcap_i\overline{E_i}]\le \Pr[\overline{E_j}]=\rel_j$ for each
$j$. The lower bound is the union bound:
$\Pr[\bigcup_i E_i]\le\sum_i(1-\rel_i)$, so
$\Pr[\bigcap_i\overline{E_i}]\ge 1-\sum_i(1-\rel_i)$. Independence gives
$\Pr[\bigcap_i\overline{E_i}]=\prod_i\Pr[\overline{E_i}]=\prod_i\rel_i$.
\end{proof}

\noindent\emph{Interpretation.} Pure decomposition into dependent steps
can only \emph{lose} reliability, and it loses geometrically with depth.
The error budget $\sum_i(1-\rel_i)$ is the right first-order accounting.
This is the problem that the other three combinators exist to solve.

\subsection{Parallel voting amplifies above one half}
Consider a decision with a unique correct answer and a per-sample
correctness probability $p$ (the binary or large-margin case).
\begin{lemma}[Vote law]\label{lem:vote}
For $n$ independent samples (take $n$ odd to avoid ties) aggregated by
majority, if $p>\tfrac12$ then
\[
\rel(\oplus_{\maj}(s,n))\;\ge\;1-\exp\!\big(-2n(p-\tfrac12)^2\big),
\]
while if $p<\tfrac12$ the same argument applied to the complement shows
the majority is correct with probability at most
$\exp(-2n(\tfrac12-p)^2)\to0$, and if $p=\tfrac12$ the votes carry no
information about the answer. For $M$-ary answers in which the correct
answer beats every alternative by expected margin $\theta>0$,
$\rel\ge 1-(M-1)\exp(-n\theta^2/2)$.
\end{lemma}
\begin{proof}
Let $X_j=\mathbf 1[\text{sample }j\text{ correct}]$, i.i.d.\ Bernoulli$(p)$.
Majority is correct iff $\bar X>\tfrac12$. By Hoeffding's
inequality~\cite{hoeffding1963},
$\Pr[\bar X\le\tfrac12]=\Pr[\bar X-p\le-(p-\tfrac12)]\le
\exp(-2n(p-\tfrac12)^2)$ when $p>\tfrac12$. For $M$-ary, apply a Hoeffding
bound to the gap between the true answer's count and each competitor's
and union-bound over the $M-1$ competitors.
\end{proof}

\noindent The phase transition at $p^\star=\tfrac12$ is Condorcet's jury
theorem~\cite{condorcet1785} in quantitative form: a committee of
better-than-chance jurors converges to truth; a committee of
worse-than-chance jurors converges to falsehood.

\subsection{Verification multiplies the odds}
Write the \emph{odds} of correctness as $\odds(s)\eqdef
\rel(s)/(1-\rel(s))$. The following lemma is the central tool of the
paper.

\begin{lemma}[Verification odds law]\label{lem:odds}
Let $g$ generate a correct candidate with probability $p$, and let $v$ be
a verifier with completeness $\beta$ and false-acceptance $\alpha>0$,
whose errors are independent of $g$'s given correctness. Condition on
acceptance. Then the accepted answer's odds of correctness satisfy
\[
\odds_{\mathrm{post}}\;=\;\odds_{\mathrm{pre}}\cdot \LR,
\qquad \LR=\frac{\beta}{\alpha},
\]
where $\odds_{\mathrm{pre}}=p/(1-p)$. Equivalently the post-acceptance
reliability is $r=\dfrac{p\beta}{p\beta+(1-p)\alpha}$.
\end{lemma}
\begin{proof}
By Bayes' rule on the event $\textsf{acc}$,
\[
\frac{\Pr[\text{corr}\mid\textsf{acc}]}{\Pr[\text{wrong}\mid\textsf{acc}]}
=\frac{\Pr[\textsf{acc}\mid\text{corr}]}{\Pr[\textsf{acc}\mid\text{wrong}]}
\cdot\frac{\Pr[\text{corr}]}{\Pr[\text{wrong}]}
=\frac{\beta}{\alpha}\cdot\frac{p}{1-p}.\qedhere
\]
\end{proof}

Two consequences are immediate and important. First, verification is a
\emph{Bayesian update}: each gate contributes additively in log-odds,
$\log\left(\odds_{\mathrm{post}}\right)=\log\left(\odds_{\mathrm{pre}}\right) + \log\left(\LR\right)$. Second, the
gain from a single gate is \emph{capped by the verifier}: one gate
multiplies the odds by at most $\LR=\beta/\alpha$, no matter how the
candidate was produced; a perfect verifier ($\alpha=0,\ \LR=\infty$)
certifies correctness outright, while an uninformative one ($\LR=1$)
changes nothing. Verification converts \emph{generation luck} into
\emph{checking power}.

\begin{lemma}[Coverage of a gate]\label{lem:cov}
The per-round acceptance probability is $q=p\beta+(1-p)\alpha$, so the
gate $V_v(g,T)$ commits with probability $1-(1-q)^T$ and the expected
number of rounds is $\E[N]=(1-(1-q)^T)/q\le 1/q$.
\end{lemma}
\begin{proof}
Rounds are i.i.d.; acceptance per round has probability $q$ by the law of
total probability. The commit probability and truncated-geometric mean
follow.
\end{proof}

\section{Amplification and the Threshold Dichotomy}\label{sec:amp}
We now stack gates. Suppose we hold a candidate and submit it to $k$
verifiers whose errors are conditionally independent given correctness,
accepting the candidate only if \emph{all} accept (re-generating
otherwise). By Lemma~\ref{lem:odds} applied $k$ times the odds multiply.

\begin{theorem}[Reliability amplification]\label{thm:amp}
Let the base generator have correctness probability $p_0\in(0,1)$ and let
$k$ conditionally independent verifiers have discriminations
$\LR_1,\dots,\LR_k$. Conditioned on joint acceptance, the reliability is
$r_k=\odds_k/(1+\odds_k)$ with $\odds_k=\frac{p_0}{1-p_0}\prod_{i=1}^k\LR_i$.
In particular, with each $\LR_i\ge\LR>1$, reliability $1-\delta$ is
attained once
\[
k\;\ge\;\frac{\log\frac{1-\delta}{\delta}+\log\frac{1-p_0}{p_0}}{\log\LR}
\;=\;O\!\Big(\tfrac{1}{\log\LR}\log\tfrac1\delta\Big),
\]
and the expected number of base invocations is of order
$k/(p_0\prod_i\beta_i)$.
\end{theorem}
\begin{proof}
Joint acceptance has likelihood $\prod_i\beta_i$ under correctness and
$\prod_i\alpha_i$ under error (conditional independence). The odds form of
Bayes gives $\odds_k=\odds_0\prod_i(\beta_i/\alpha_i)=\odds_0\prod_i\LR_i$.
Solving $r_k\ge 1-\delta\iff \odds_k\ge(1-\delta)/\delta$ and taking
logarithms yields the stated $k$. For cost, by Lemma~\ref{lem:cov} the
per-attempt joint-acceptance probability is at least
$p_0\prod_i\beta_i$, so the expected number of generate-and-check
attempts is $O(1/(p_0\prod_i\beta_i))$, each costing $k{+}1$ invocations.
\end{proof}

The depth is \emph{logarithmic} in the target error $\delta$: closing the
reliability gap is exponentially cheap in structure, provided the
verifier is informative. The next theorem shows that proviso is exactly a
phase boundary.

\begin{theorem}[Threshold dichotomy]\label{thm:threshold}
Fix a base solver and consider amplifying its reliability with bounded
per-stage cost.
\begin{enumerate}
\item \emph{(Verification.)} If $\LR>1$, then for every $\delta>0$ there
is a strategy of depth $O(\log\frac1\delta)$ and reliability $\ge1-\delta$.
If $\LR=1$, then for \emph{every} verification strategy
$\rel\le p_0$; the verifier adds nothing. If $\LR<1$ (that is,
$\beta<\alpha$), swapping the verifier's accept and reject verdicts gives
discrimination $(1-\beta)/(1-\alpha)>1$, so amplification is again
possible and $\LR^\star=1$ is the sole critical value.
\item \emph{(Voting.)} If $p>\tfrac12$, majority voting reaches any
$1-\delta$ with $n=O\!\big(\frac{1}{(p-\frac12)^2}\log\frac1\delta\big)$
samples; if $p\le\tfrac12$, $\rel(\oplus_{\maj}(s,n))$ does not exceed
$p$ and tends to $0$ for $p<\tfrac12$. The critical value is
$p^\star=\tfrac12$.
\end{enumerate}
\end{theorem}
\begin{proof}
Part (1), $\LR>1$: Theorem~\ref{thm:amp}. Part (1), $\LR=1$: then
$\beta=\alpha$, so acceptance is independent of correctness; conditioning
on $\textsf{acc}$ leaves $\Pr[\text{corr}\mid\textsf{acc}]=p_0$ by
Lemma~\ref{lem:odds} with $\LR=1$, and composing such gates preserves the
posterior, so no strategy built only from uninformative gates exceeds
$p_0$. $\LR<1$: swap the verdicts as in the statement. Part (2): the upper
direction is Lemma~\ref{lem:vote}; for $p<\tfrac12$, $\bar X\to
p<\tfrac12$ a.s.\ by the law of large numbers, so majority is eventually
always wrong and $\rel\to0$; for $p=\tfrac12$ the votes are independent
of the truth, so no aggregation rule can do better than probability
$\tfrac12$.
\end{proof}

This is the problem-solving counterpart of the fault-tolerance
\emph{threshold theorems} for unreliable computation and quantum
error-correction~\cite{vonneumann1956,aharonov1997}: there is a critical
component quality above which arbitrarily reliable computation is
achievable at modest overhead, and below which it is not. Here the
``component quality'' that matters for checking is the verifier's
likelihood ratio, and for voting it is being better than chance.
Figure~\ref{fig:amp} plots the two regimes.

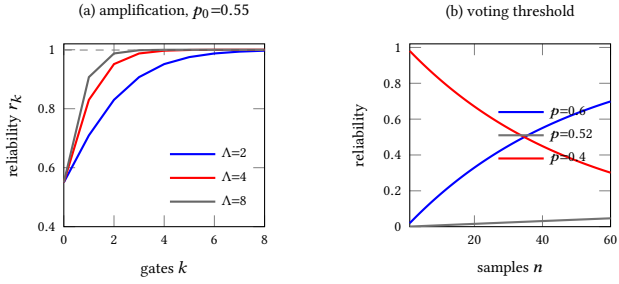
\begin{figure}[t]
\centering
\begin{tikzpicture}
\begin{axis}[
  width=0.50\columnwidth,height=4.0cm,
  xlabel={\scriptsize gates $k$},ylabel={\scriptsize reliability $r_k$},
  xmin=0,xmax=8,ymin=0.4,ymax=1.02,
  tick label style={font=\tiny},label style={font=\scriptsize},
  legend style={font=\tiny,at={(0.98,0.04)},anchor=south east,draw=none,fill=none},
  legend cell align=left,title style={font=\scriptsize},title={(a) amplification, $p_0{=}0.55$}]
\addplot[blue,thick,domain=0:8,samples=9] {(1.2222*(2)^x)/(1+1.2222*(2)^x)};
\addlegendentry{$\LR{=}2$}
\addplot[red,thick,domain=0:8,samples=9] {(1.2222*(4)^x)/(1+1.2222*(4)^x)};
\addlegendentry{$\LR{=}4$}
\addplot[black!60,thick,domain=0:8,samples=9] {(1.2222*(8)^x)/(1+1.2222*(8)^x)};
\addlegendentry{$\LR{=}8$}
\addplot[gray,dashed,domain=0:8] {0.999};
\end{axis}
\begin{scope}[xshift=0.54\columnwidth]
\begin{axis}[
  width=0.50\columnwidth,height=4.0cm,
  xlabel={\scriptsize samples $n$},ylabel={\scriptsize reliability},
  xmin=1,xmax=60,ymin=0,ymax=1.02,
  tick label style={font=\tiny},label style={font=\scriptsize},
  legend style={font=\tiny,at={(0.98,0.5)},anchor=east,draw=none,fill=none},
  legend cell align=left,title style={font=\scriptsize},title={(b) voting threshold}]
\addplot[blue,thick,domain=1:60,samples=60]{1-exp(-2*x*(0.10)^2)};
\addlegendentry{$p{=}0.6$}
\addplot[black!55,thick,domain=1:60,samples=60]{1-exp(-2*x*(0.02)^2)};
\addlegendentry{$p{=}0.52$}
\addplot[red,thick,domain=1:60,samples=60]{exp(-2*x*(0.10)^2)};
\addlegendentry{$p{=}0.4$}
\end{axis}
\end{scope}
\end{tikzpicture}
\caption{(a) Verification amplifies reliability geometrically in the
number of gates; larger discrimination $\LR$ reaches the $1{-}\delta$
target (dashed) in fewer gates (Theorem~\ref{thm:amp}). (b) Majority
voting amplifies above $p^\star{=}\tfrac12$ and \emph{degrades} below it
(Theorem~\ref{thm:threshold}); the curve for $p{=}0.52$ shows how slowly
near-threshold solvers improve.}
\label{fig:amp}
\end{figure}

\section{Self-Organization as a Fixed Point}\label{sec:lattice}
We now explain why structure should \emph{appear}: why a system that
locally improves itself converges to a stable, layered organization.

\paragraph{The strategy lattice.}
Let $\mathbb{S}$ be the set of strategies for a fixed problem family,
augmented with a bottom element $\bot_{\mathbb S}$ (the trivial
abstaining solver) and a top element $\top_{\mathbb S}$ (the oracle).
Define the \emph{refinement order} $\sigma\sqsubseteq\sigma'$ to mean
$\sigma'$ is obtained from $\sigma$ by a finite sequence of
\emph{reliability-monotone expansions}: wrapping a subtree in a vote
(only where the subtree is above chance), in a gate (only with
$\LR\ge1$), or in an additional refinement round, or replacing a leaf by
a decomposition whose composite dominates it pointwise. The qualifiers
matter: by the laws of \S\ref{sec:laws}, these are exactly the conditions
under which each expansion cannot decrease reliability.

\begin{lemma}[Complete lattice]\label{lem:lattice}
$(\mathbb{S},\sqsubseteq)$ is a complete lattice: every subset has a
supremum (the join obtained by parallel ensembling with a
correctness-selecting aggregator) and an infimum.
\end{lemma}
\begin{proof}[Proof sketch]
Finite joins exist by $\oplus$ with an idealized selector that returns a
correct member if any branch is correct, which dominates each branch in
$\sqsubseteq$; finite meets exist dually by restricting coverage to the
common domain. Arbitrary joins/meets exist by Dedekind--MacNeille
completion of the resulting poset, which adjoins all suprema and infima
while preserving existing ones~\cite{daveypriestley2002}. The adjoined
$\top_{\mathbb S},\bot_{\mathbb S}$ are the global bounds. (The join is
an idealization used only to give the order a complete-lattice structure;
the improvement operator below never needs to construct it.)
\end{proof}

\paragraph{The improvement operator.}
Fix a cost budget $\kappa$ and a marginal-rate threshold $\lambda>0$.
Define $\Phi_\kappa:\mathbb{S}\to\mathbb{S}$ as the operator that, given a
strategy $\sigma$, applies the single combinator (add a gate, add a vote,
deepen a decomposition) that maximizes the \emph{marginal log-odds gain
per unit cost}, provided that rate exceeds $\lambda$ and the budget
$\kappa$ is not exhausted; otherwise $\Phi_\kappa(\sigma)=\sigma$. (Ties
in the maximization are broken by a fixed priority order, so
$\Phi_\kappa$ is a well-defined function.)

\begin{algorithm}[t]
\caption{\textsc{SelfOrganize}: iterate the improvement operator}
\label{alg:selforg}
\begin{algorithmic}[1]
\State \textbf{input:} seed strategy $\sigma_0$, budget $\kappa$, rate
threshold $\lambda$, candidate expansions $\mathcal{E}$
\State $\sigma\gets\sigma_0$
\While{$\cost(\sigma)<\kappa$}
  \State $e^\star\gets\argmax_{e\in\mathcal{E}(\sigma)}
         \dfrac{\Delta\log\odds(e)}{\Delta\cost(e)}$
         \Comment{best marginal log-odds per cost}
  \If{$\dfrac{\Delta\log\odds(e^\star)}{\Delta\cost(e^\star)}\le\lambda$}
     \State \textbf{break} \Comment{verifier-saturated: no rate exceeds $\lambda$}
  \EndIf
  \State $\sigma\gets e^\star(\sigma)$
         \Comment{apply gate / vote / decomposition (monotone)}
\EndWhile
\State \textbf{return} $\sigma$
\end{algorithmic}
\end{algorithm}

Algorithm~\ref{alg:selforg} realizes $\Phi_\kappa$ as hill-climbing on the
strategy lattice. Its termination and the structure of its output are not
implementation details but theorems:

\begin{lemma}[Monotonicity]\label{lem:mono}
Under the reliability-monotone expansion order, $\Phi_\kappa$ is
order-preserving: $\sigma\sqsubseteq\sigma'\Rightarrow
\Phi_\kappa(\sigma)\sqsubseteq\Phi_\kappa(\sigma')$.
\end{lemma}
\begin{proof}[Proof sketch]
Each admissible expansion is, by Lemmas~\ref{lem:vote}--\ref{lem:odds},
reliability non-decreasing and acts on a subtree; applying the same class
of expansion to a refinement $\sigma'$ of $\sigma$ yields a refinement of
$\Phi_\kappa(\sigma)$ because refinement is preserved under wrapping in a
gate/vote and under dependent substitution. Hence the image order is
preserved.
\end{proof}

\begin{theorem}[Existence of a self-organized strategy]\label{thm:fixed}
$\Phi_\kappa$ has a least fixed point above any seed $\sigma_0$ and a
greatest fixed point below $\top_{\mathbb S}$. Moreover, the iteration
$\sigma_0$, $\Phi_\kappa(\sigma_0)$, $\Phi_\kappa^{2}(\sigma_0)$, \dots\
reaches
the least fixed point above $\sigma_0$, denoted $\sigma^\star_\kappa$,
after finitely many steps.
\end{theorem}
\begin{proof}
By Lemma~\ref{lem:lattice} the domain is a complete lattice and by
Lemma~\ref{lem:mono} $\Phi_\kappa$ is monotone, so the Knaster--Tarski
theorem~\cite{tarski1955} (every order-preserving map on a complete
lattice has least and greatest fixed points) applies. For the
iteration: every step that changes the strategy adds at least the cost
$c_{\min}>0$ of the cheapest combinator, and total cost is capped by
$\kappa$, so after at most $\lceil\kappa/c_{\min}\rceil$ steps the
iteration is constant at some fixed point
$\sigma^\star_\kappa\sqsupseteq\sigma_0$. It is the least such: if $\tau$
is any fixed point with $\tau\sqsupseteq\sigma_0$, monotonicity gives
$\Phi_\kappa^{\,n}(\sigma_0)\sqsubseteq\Phi_\kappa^{\,n}(\tau)=\tau$ for
all $n$, hence $\sigma^\star_\kappa\sqsubseteq\tau$.
\end{proof}

\begin{proposition}[Characterization: verifier-saturation]\label{prop:sat}
At the fixed point $\sigma^\star_\kappa$, no admissible combinator yields
marginal log-odds gain per unit cost exceeding $\lambda$. Consequently the
organization is \emph{verifier-saturated}: under the amplification law
(Theorem~\ref{thm:amp}), additional gates have been added until their
marginal rate $\log\LR/\Delta\cost$ falls to $\lambda$, and the
optimal allocation of a fixed budget across $J$ candidate expansions
equalizes the marginal rates,
\[
\frac{\partial \log\odds}{\partial \cost}\Big|_{j}=\lambda
\quad\text{for all active }j,
\]
a \emph{water-filling} condition: budget flows to whichever expansions
currently offer the highest return, until all active ones offer the same
rate $\lambda$. (This is the standard first-order condition for
constrained maximization.)
\end{proposition}
\begin{proof}
At a fixed point $\Phi_\kappa(\sigma)=\sigma$, so by definition of
$\Phi_\kappa$ no available expansion has rate $>\lambda$ within budget;
hence every active expansion sits at rate $=\lambda$ (lower-rate ones are
inactive), which is the stationarity (KKT) condition of maximizing
$\log\odds$ subject to $\cost\le\kappa$ with multiplier $\lambda$.
\end{proof}

Two messages follow. First, \emph{organization is emergent and inevitable
under local improvement}: any hill-climbing process on the strategy
lattice that prefers higher reliability-per-cost terminates at a layered,
verifier-saturated structure; it does not need a designer. Second, the
shape of that structure is dictated by the laws of \S\ref{sec:laws}: it
allocates checking where information is cheapest, exactly as a rational
designer would, because the fixed point coincides with the constrained
optimum.

\section{Fundamental Limits}\label{sec:lower}
Amplification is bounded by information, by correlation, and by the
absence of free lunches.

\subsection{An information ceiling on verification}
A verifier is a channel from the latent bit ``correct/incorrect'' to its
verdict. Its discrimination cannot exceed the information that channel
carries.

\begin{theorem}[Information ceiling]\label{thm:info}
Let a verifier output verdict $W$ from latent correctness $C\in\{0,1\}$.
The expected log-odds gain it produces (averaged over the verdict, given
a correct candidate) is the divergence between its two verdict
distributions:
\[
\E\big[\log\odds_{\mathrm{post}}-\log\odds_{\mathrm{pre}}\big]
=\DKL\!\big(P_{W\mid C=1}\,\|\,P_{W\mid C=0}\big),
\]
and no cascade of verifiers that are all functions of the same evidence
$Z$ can yield a more reliable decision than the best decision based on
$Z$ itself. In particular, by the data-processing inequality, achievable
reliability is limited by the mutual information $I(C;Z)$ between
correctness and all evidence the system can observe.
\end{theorem}
\begin{proof}
The posterior log-odds update on verdict $W$ is the log-likelihood ratio
$\log\frac{P_{W\mid C=1}(W)}{P_{W\mid C=0}(W)}$; its expectation when
$W\sim P_{W\mid C=1}$ is, by definition,
$\DKL(P_{W\mid1}\|P_{W\mid0})$. Verdicts are (possibly randomized)
functions of $Z$, so by the data-processing
inequality~\cite{cover2006} $I(C;\text{verdicts})\le I(C;Z)$, and a
decision based on a function of $Z$ cannot be more accurate than the
best decision based on $Z$. Fano's inequality then bounds achievable
reliability in terms of $I(C;Z)$.
\end{proof}

Stacking gates (Theorem~\ref{thm:amp}) therefore pays off only when each
gate contributes \emph{new}, conditionally independent evidence; gates
that re-read the same signal are redundant, and their effective $\LR$
collapses toward $1$. Amplification consumes information.

\subsection{Diversity is necessary}
Independence was assumed in Lemmas~\ref{lem:vote}--\ref{lem:odds}. Real
populations of solvers often make errors \emph{together}: they share
training, methods, or blind spots. We now show exactly when shared errors
put a hard floor under majority voting.

We model shared causes in the standard way: the voters are
\emph{conditionally i.i.d.}\ given a latent factor $S$ that collects
everything they have in common. (By de Finetti's theorem, every infinite
exchangeable population has this form, so this is not a special
assumption.) Write $p(S)$ for the per-voter accuracy conditional on $S$,
with $\E[p(S)]=p$, $\sigma^2=p(1-p)$, and pairwise correlation
$\gamma=\operatorname{Var}(p(S))/\sigma^2$.

\begin{theorem}[Shared-cause voting floor]\label{thm:diversity}
In the latent-factor model, with $n$ odd:
(i) the vote share satisfies
$\operatorname{Var}(\bar X)=\frac{\sigma^2}{n}+\frac{n-1}{n}\gamma\sigma^2
\xrightarrow{n\to\infty}\gamma\sigma^2$, so for $\gamma>0$ it never
concentrates at $p$; and
(ii) if $\Pr[p(S)=\tfrac12]=0$, then
\[
\lim_{n\to\infty}\Pr[\text{majority wrong}]
\;=\;\Pr\!\big[p(S)<\tfrac12\big].
\]
The floor on the right is strictly positive exactly when shared
circumstances can push the whole committee below chance, and no amount
of replication removes it. In particular, in the Gaussian-copula model
(the model used in the companion simulations) the floor is strictly
positive for every $\gamma>0$.
\end{theorem}
\begin{proof}
(i) is the variance expansion of $\operatorname{Var}(\frac1n\sum X_j)$
with equal pairwise covariances
$\operatorname{Cov}(X_i,X_j)=\operatorname{Var}(p(S))=\gamma\sigma^2$.
(ii) Conditional on $S$, the law of large numbers gives $\bar X\to p(S)$
almost surely, so the majority indicator $\mathbf 1[\bar X>\tfrac12]$
converges to $\mathbf 1[p(S)>\tfrac12]$ outside the null event
$\{p(S)=\tfrac12\}$; bounded convergence yields the limit. For the
Gaussian-copula model,
$X_j=\mathbf 1[\sqrt{\gamma}\,S+\sqrt{1-\gamma}\,\varepsilon_j\le u]$
with independent standard normal $S,\varepsilon_j$ and $u=\Phi^{-1}(p)$,
where $\Phi$ is the standard normal distribution function; then
$p(S)=\Phi\big((u-\sqrt{\gamma}\,S)/\sqrt{1-\gamma}\big)$ takes every
value in $(0,1)$ with positive density, so
$\Pr[p(S)<\tfrac12]>0$ for every $\gamma>0$.
\end{proof}

The characterization has two sides. If shared factors only add noise but
never drag the committee below chance ($p(S)>\tfrac12$ almost surely),
majority voting still converges to the truth; correlation merely slows
it down. But if some situations make the whole committee wrong together,
those situations are lost no matter how many copies vote.

\begin{corollary}[Effective committee size]\label{cor:neff}
Matching the vote-share variance, correlated voting behaves like
independent voting with an \emph{effective} sample size
$n_{\mathrm{eff}}=\frac{n}{1+(n-1)\gamma}\to\frac1\gamma$: in terms of
how sharply its vote share concentrates, a correlated committee is never
worth more than $1/\gamma$ independent voters, however many members it
has.
\end{corollary}
\begin{proof}
Match variances: $\operatorname{Var}(\bar X)=\sigma^2/n_{\mathrm{eff}}$
with the expression of Theorem~\ref{thm:diversity} gives
$n_{\mathrm{eff}}=n/(1+(n-1)\gamma)$, whose limit is $1/\gamma$.
\end{proof}

For example, at correlation $\gamma{=}0.1$ even a thousand-member
committee concentrates no better than about ten independent voters, so
the value of one more correlated member is already negligible.
\emph{Diversity}, meaning conditionally independent errors, is therefore not a
luxury but a requirement for unbounded amplification. Organizations that
clone a single fallible solver inherit its blind spots no matter how
large they grow; useful committees are built from members who fail
differently.

\subsection{No free lunch}
\begin{corollary}[No universal decomposition]\label{cor:nfl}
Averaged uniformly over all problem families with a given answer space,
every strategy in $\Alg(B)$ has the same expected reliability as its base
solvers: decomposition confers no advantage absent a prior matching
structure to problems.
\end{corollary}
\begin{proof}[Proof sketch]
This is the no-free-lunch theorem for search and
optimization~\cite{wolpert1997} transported to our setting: combinators
are deterministic re-wirings of base-solver calls, so over the uniform
mixture of oracles every fixed strategy induces the same marginal
distribution over (instance, answer) correctness as calling the base
solvers directly. Gains on a subfamily are exactly offset elsewhere.
\end{proof}

The corollary is not nihilistic; it is a statement about \emph{where}
reliability comes from. The combinators do not manufacture reliability;
they \emph{transport} the information already present in solvers and
verifiers to the place a decision is made, and they only help on the
structured problem families we actually face: those where verifiers are
informative and errors are diverse.

\section{The Cost--Reliability Frontier}\label{sec:frontier}
Collecting the laws, we can compare combinators on a common axis: the
cost (base invocations) to reach error $\delta$.

\begin{table}[t]
\centering\small
\caption{Work needed to reach error $\delta$ and amplification behaviour
of each combinator ($p$: base correctness; $\LR$: verifier
discrimination). Verification work counts gates; with an incomplete
verifier ($\beta<1$), regenerations add a factor of at most $\beta^{-k}$,
negligible for $\beta$ near $1$.}
\label{tab:frontier}
\begin{tabular}{@{}lll@{}}
\toprule
Combinator & Work for error $\delta$ & Regime \\
\midrule
Sequential $\seq$ & --- & anti-amplifying ($\rel\!\downarrow$) \\
Voting $\oplus_{\maj}$ & $\Theta\!\big(\frac{1}{(p-\frac12)^2}\log\frac1\delta\big)$ & needs $p>\tfrac12$ \\
Verify $V_v$ & $\Theta\!\big(\frac{1}{\log\LR}\log\frac1\delta\big)$ & needs $\LR>1$ \\
Recurse $\mu$ & depth-dependent & inherits children \\
\bottomrule
\end{tabular}
\end{table}

\begin{proposition}[Log-optimality and allocation]\label{prop:frontier}
Both voting and verification reach error $\delta$ with structure size
$\Theta(\log\frac1\delta)$ (samples for voting, gates for verification),
which is optimal for any method whose error decays at most geometrically
per unit of structure; verification's constant is governed by $\log\LR$
and voting's by $(p-\tfrac12)^2$, so verification dominates whenever a
sufficiently discriminating checker exists. Under a
fixed budget the reliability-maximizing strategy equalizes marginal
log-odds gain per cost across stages (Proposition~\ref{prop:sat}),
preferring the combinator with the largest current marginal rate.
\end{proposition}
\begin{proof}
Each accepted gate adds $\log\LR$ to the log-odds, and each batch of $n$
votes adds $\Theta((p-\tfrac12)^2 n)$ to a Chernoff exponent; both give
$\log(1/\delta)$ scaling, and error that decays at most geometrically per
unit of structure needs $\Omega(\log\frac1\delta)$ structure. The
allocation rule is the first-order optimality (KKT) condition of
Proposition~\ref{prop:sat}. In total calls, verification additionally
pays a regeneration factor of at most $\beta^{-k}$, which equals $1$ for
a complete verifier and stays small for $\beta$ near $1$.
\end{proof}

\paragraph{A worked organization.}
Suppose a base solver with $p_0=0.55$ and a verifier with $\LR=4$
($\beta=0.8,\alpha=0.2$). Each gate multiplies the odds by $4$; the
log-odds grow linearly while reliability saturates toward one
(Table~\ref{tab:calc}). Reaching $\delta=10^{-3}$ needs
$k=\lceil(\log 999+\log\frac{0.45}{0.55})/\log4\rceil=5$ gates, matching
Theorem~\ref{thm:amp}. Pure voting from $p_0=0.55$ would instead need
$\sim\!\frac{1}{2(0.05)^2}\log10^3\approx1400$ samples for the same
target: counting all calls (including regenerations), verification is
roughly two orders of magnitude cheaper here (a few dozen calls versus
about $1400$), which is exactly why effective organizations are built
around checkers.

\begin{table}[t]
\centering\small
\caption{Reliability calculus for the worked organization: $k$ gates,
$p_0{=}0.55$, $\LR{=}4$. Odds multiply; reliability saturates.}
\label{tab:calc}
\begin{tabular}{@{}cccc@{}}
\toprule
gates $k$ & odds $\odds_k$ & reliability $r_k$ & error $1-r_k$ \\
\midrule
0 & 1.22 & 0.550 & $4.5\times10^{-1}$ \\
1 & 4.89 & 0.830 & $1.7\times10^{-1}$ \\
2 & 19.6 & 0.951 & $4.9\times10^{-2}$ \\
3 & 78.2 & 0.987 & $1.3\times10^{-2}$ \\
4 & 313  & 0.997 & $3.2\times10^{-3}$ \\
5 & 1252 & 0.9992 & $8.0\times10^{-4}$ \\
\bottomrule
\end{tabular}
\end{table}

\section{Recursive Decomposition: A Reliability Master Theorem}\label{sec:master}
The recursion combinator $\mu$ raises the stakes: a problem is reduced to
subproblems, each solved by the \emph{same} repertoire, to a depth $d$.
Without correction, recursion is the serial law (Lemma~\ref{lem:seq})
compounded across an exponentially growing tree, and reliability
collapses. With verification at each node it is reliable at polylogarithmic
overhead. We make this precise.

Consider a balanced recursion of depth $d$ and branching $b$: an internal
node reduces its instance to $b$ subinstances, solves each recursively,
and recombines with reliability $\rho_c$ (the recombiner is correct, given
correct children, with probability $\rho_c$); a leaf is solved by a base
solver of reliability $p$.

\begin{proposition}[Unverified recursion collapses]\label{prop:collapse}
Without per-node verification, the reliability $r(d)$ of the recursive
solver obeys $r(d)=\rho_c\,r(d-1)^{b}$ with $r(0)=p$, whose closed form
is
\[
r(d)\;=\;\rho_c^{\,(b^d-1)/(b-1)}\;p^{\,b^d}.
\]
Hence for $b\ge2$, unless $p=\rho_c=1$, $r(d)\to0$ as $d\to\infty$, and
the decay is doubly exponential in the depth $d$.
\end{proposition}
\begin{proof}
A node is correct iff its recombiner is correct and all $b$ children are
correct; by independence $r(d)=\rho_c\,r(d-1)^{b}$. Unrolling the
recurrence gives
$r(d)=\rho_c^{\,1+b+\cdots+b^{d-1}}p^{\,b^d}
=\rho_c^{\,(b^d-1)/(b-1)}p^{\,b^d}$. Both exponents grow like $b^d$, so
$r(d)=c^{\,b^d(1+o(1))}$ with $c=p\,\rho_c^{1/(b-1)}<1$ unless
$p=\rho_c=1$.
\end{proof}

This is why naive ``decompose-and-recurse'' degrades with depth: each
level multiplies the number of opportunities for error by the branching
factor $b$. The remedy is to restore each node's reliability before its
result propagates upward.

\begin{theorem}[Recursion master theorem]\label{thm:master}
Let the tree have $N=\frac{b^{d+1}-1}{b-1}$ nodes, and suppose each node
is wrapped in a verification gate that drives its local error (given
correct inputs) to at most $\eta$. Then the whole computation is correct
with probability at least $1-N\eta$. Consequently, to achieve overall
error $\delta$ it suffices to verify each node to error
$\eta=\delta/N$, which by Theorem~\ref{thm:amp} requires per-node
verification depth
\[
k_{\mathrm{node}}=O\!\Big(\tfrac{1}{\log\LR}\big(\log\tfrac1\delta+d\log b\big)\Big),
\]
and total cost $O\!\big(N\,k_{\mathrm{node}}\big)=
O\!\big(b^{d}\,(\log\tfrac1\delta+d\log b)/\log\LR\big)$: linear in the
work $b^d$ up to a polylogarithmic overhead, provided $\LR>1$.
\end{theorem}
\begin{proof}
Let $A_v$ be the event that node $v$'s local step errs despite correct
inputs, $\Pr[A_v]\le\eta$. The output is correct unless some node errs, so
by the union bound the failure probability is at most $\sum_v\Pr[A_v]\le
N\eta$. Setting $N\eta=\delta$ gives $\eta=\delta/N$; achieving per-node
error $\delta/N$ from a fixed base by Theorem~\ref{thm:amp} costs
$O(\log(N/\delta)/\log\LR)$ gates, and $\log N=\Theta(d\log b)$. Summing
over $N$ nodes gives the total. The count is in gate evaluations; with an
incomplete verifier ($\beta<1$), regenerations multiply per-node cost by
at most $\beta^{-k_{\mathrm{node}}}=(N/\delta)^{\ln(1/\beta)/\ln\LR}$, a
mild polynomial factor that equals $1$ for a complete checker
($\beta=1$, e.g.\ a proof checker).
\end{proof}

\begin{corollary}[Threshold for recursion]\label{cor:rec-threshold}
Recursive decomposition is reliability-preserving to arbitrary depth iff
an informative verifier exists ($\LR>1$). At $\LR=1$ no per-node
correction is possible and reliability follows the collapsing recurrence
of Proposition~\ref{prop:collapse}.
\end{corollary}

\paragraph{A worked recursion.}
Take depth $d{=}4$, branching $b{=}3$ (so $N=\frac{3^5-1}{2}=121$ nodes),
target $\delta{=}10^{-2}$, and a robust verifier with $\LR{=}4$. The
master theorem asks each node to reach error
$\eta=\delta/N\approx8.3\times10^{-5}$, i.e.\ odds
$\approx1.2\times10^{4}$, which from $p_0{=}0.55$ takes
$k_{\mathrm{node}}=\lceil(\log(1.2{\times}10^4)+\log\frac{0.45}{0.55})/\log4\rceil
=7$ gates. The unverified recurrence (Prop.~\ref{prop:collapse}) with the
same leaves and $\rho_c{=}0.9$ instead falls to $0.15$ after one level,
$0.003$ after two, and below $10^{-7}$ after three, a clear illustration
that depth without checking destroys reliability, while depth with
logarithmic-overhead checking preserves it.

The master theorem is the recursion-level counterpart of the threshold
dichotomy and of fault-tolerant computing: \emph{depth is affordable when
each level is checked}. The polylogarithmic-per-node overhead is precisely
the price of preventing error from compounding across the tree; it has
the same shape as the overhead in fault-tolerant circuits, but is derived
here from the odds law rather than from code distance.

\section{Selective Reliability: Abstention and Calibration}\label{sec:abstain}
So far reliability was measured on \emph{committed} answers
(Definition~\ref{def:rel}). The freedom to abstain ($\bot$) is itself a
reliability lever: a system can decline the instances it is least sure of
and raise its reliability on the rest. This is selective
prediction~\cite{chow1970,elyaniv2010}, and it interacts cleanly with our
laws because the verifier already supplies a confidence signal: its
log-likelihood-ratio score.

\begin{definition}[Risk--coverage]
For a score $\phi(x,a)$ (e.g.\ the accumulated log-odds after gating) and
threshold $\tau$, the selective solver commits iff $\phi\ge\tau$. Its
coverage is $\mathrm{cov}(\tau)=\Pr[\phi\ge\tau]$ and its risk is
$R(\tau)=\Pr[a\notin Y^\star\mid \phi\ge\tau]$. The \emph{risk--coverage
curve} is $\tau\mapsto(\mathrm{cov}(\tau),R(\tau))$.
\end{definition}

\begin{proposition}[Abstention monotonicity and optimal threshold]\label{prop:chow}
If $\phi$ is the true log-odds of correctness (a calibrated score), then
$R(\tau)$ is non-increasing in $\tau$: tightening the acceptance threshold
never worsens reliability on the covered region. Moreover, under costs
$c_{\mathrm{err}}$ for a wrong commit and $c_{\mathrm{abs}}$ for
abstaining, the risk-minimizing policy is Chow's rule: commit iff the
posterior correctness exceeds $1-c_{\mathrm{abs}}/c_{\mathrm{err}}$,
equivalently $\phi\ge\log\frac{c_{\mathrm{err}}-c_{\mathrm{abs}}}{c_{\mathrm{abs}}}$.
\end{proposition}
\begin{proof}
For calibrated $\phi$, $\Pr[\text{wrong}\mid\phi]=\sigma(-\phi)$ is
decreasing in $\phi$, so the conditional risk over the region
$\{\phi\ge\tau\}$ is an average of decreasing tails and is non-increasing
in $\tau$. The Bayes-optimal decision compares expected costs
$c_{\mathrm{err}}\,\Pr[\text{wrong}\mid\phi]$ against $c_{\mathrm{abs}}$;
commit when the former is smaller, which rearranges to the stated
log-odds threshold (Chow's rule~\cite{chow1970}).
\end{proof}

Thus the verifier does double duty: as a \emph{gate} it amplifies
reliability (\S\ref{sec:amp}); as a \emph{score} it orders instances for
selective abstention (here). The risk--coverage curve is the achievable
frontier of an organization, and escalation (abstaining and routing hard
instances to a stronger sub-organization) is just abstention with a
fallback, inheriting these guarantees.

\section{Generator--Verifier Games}\label{sec:games}
The amplification theorem assumed the generator's errors are
\emph{indifferent} to the verifier. When a generator instead
\emph{optimizes} for acceptance, as learned systems do under a reward and
as nature does under selection, the verifier's discrimination can
silently collapse. This is Goodhart's law (a proxy measure stops
tracking what it was meant to measure once it is optimized directly)
expressed in our calculus, and it limits how much reliability
organization can manufacture.

Model a Stackelberg interaction: the verifier $v$ is fixed (the leader);
the generator chooses an output distribution $g$ (the follower) to
maximize acceptance probability, possibly concentrating mass on
\emph{wrong} answers that $v$ nonetheless accepts. Let
$\alpha^\dagger(v)=\sup_{g:\,g\text{ wrong}}\Pr[v\ \textsf{acc}\mid g]$ be
the \emph{adversarial} false-acceptance rate.

\begin{definition}[Non-gameable verifier]
A verifier is $\Lambda$-\emph{robust} if its discrimination stays above
$\Lambda>1$ against the worst-case generator, i.e.\
$\beta/\alpha^\dagger(v)\ge\Lambda$.
\end{definition}

\begin{theorem}[Robust amplification]\label{thm:robust}
If every gate is $\Lambda$-robust with $\Lambda>1$, the amplification and
master theorems (Thms.~\ref{thm:amp},~\ref{thm:master}) hold verbatim
against adversarial generation, with $\LR$ replaced by the robust
$\Lambda$. If some gate has $\alpha^\dagger=\beta$ (fully gameable), its
effective $\LR\to1$ and it contributes no amplification regardless of its
nominal $\beta/\alpha$ on benign inputs.
\end{theorem}
\begin{proof}
The odds law (Lemma~\ref{lem:odds}) used only $\Pr[\textsf{acc}\mid
\text{corr}]$ and $\Pr[\textsf{acc}\mid\text{wrong}]$; substituting the
worst-case wrong-acceptance $\alpha^\dagger$ for $\alpha$ yields posterior
odds multiplied by $\beta/\alpha^\dagger\ge\Lambda$, and the proofs of
Theorems~\ref{thm:amp} and~\ref{thm:master} go through with $\Lambda$. If
$\alpha^\dagger=\beta$ the multiplier is $1$, so by the $\LR{=}1$ case of
Theorem~\ref{thm:threshold} no amplification occurs.
\end{proof}

The lesson sharpens \S\ref{sec:lower}: it is not enough for a verifier to
be discriminating on average; to support reliable organization under
optimization pressure it must be discriminating against an adversary that
\emph{searches} for what it will wrongly accept. Proof checkers and type
systems are robust in this sense (their false-acceptance is bounded by
soundness, independent of how the prover was chosen); learned reward
models and shallow heuristics often are not, which is precisely why
reliability built on them erodes as generators are optimized against
them. Robust verification is the scarce resource on which reliable
intelligence is built.

\section{Related Work}\label{sec:related}
\paragraph{Reliable computation from unreliable parts.}
Von Neumann's multiplexing and the subsequent fault-tolerance threshold
theorems~\cite{vonneumann1956,aharonov1997} established that redundancy
plus error-correction yields arbitrarily reliable computation above a
component-quality threshold. Our Theorem~\ref{thm:threshold} is the
analogue for problem-solving, with the verifier's likelihood ratio in the
role of component quality and verification in the role of correction.

\paragraph{Aggregation and ensembles.}
The Condorcet jury theorem~\cite{condorcet1785}, boosting and ensemble
learning~\cite{schapire1990}, and concentration of
measure~\cite{hoeffding1963} explain when many weak deciders make a strong
one; our vote law and diversity floor (Theorems~\ref{lem:vote},
\ref{thm:diversity}) are the reliability-calculus form, emphasizing that
correlation caps the gain.

\paragraph{Verification and learning theory.}
The asymmetry that checking can be easier than producing underlies
complexity theory's verifier-centric definitions and PAC
learning~\cite{valiant1984}; Lemma~\ref{lem:odds} quantifies the value of
a checker as a likelihood ratio and connects to information-theoretic
limits~\cite{cover2006}. Selective prediction and the reject
option~\cite{chow1970,elyaniv2010} furnish the abstention frontier of
\S\ref{sec:abstain}, and our generator--verifier analysis (\S\ref{sec:games})
formalizes the Goodhart effect that erodes non-robust checkers under
optimization pressure.

\paragraph{Compositional and categorical models.}
Markov categories give a clean syntax for stochastic
composition~\cite{fritz2020,chojacobs2019}; lattice and fixed-point
methods underlie program semantics~\cite{tarski1955,daveypriestley2002}.
We use both to make ``organization'' a formal object and to prove that
self-organization is a fixed point.

\paragraph{Problem-solving architectures.}
From the General Problem Solver and means--ends analysis~\cite{newell1972}
to the society-of-mind view that intelligence is the organization of
many small unintelligent processes~\cite{minsky1986}, and the
divide-and-conquer paradigm in algorithms and distributed
systems~\cite{dean2004}, the recurring idea is that capability is a
property of structure. We supply a reliability-theoretic account of which
structures work and why.

\section{Discussion and Conclusion}\label{sec:concl}
We modeled problem-solving as an algebra over unreliable solvers and
found that a small generating set (sequence, vote, verify, recurse)
suffices to express the organizations intelligence actually
uses, and that reliability flows through these combinators by explicit
laws. Sequencing degrades; voting amplifies above chance; verification
multiplies the odds by the checker's likelihood ratio and so, stacked,
closes the error gap at logarithmic depth. Above critical parameters
($\LR^\star=1$, $p^\star=\tfrac12$) reliability is essentially free to
buy; at or below them it cannot be bought at all. Crucially, the layered,
verifier-saturated organizations we observe in capable systems are not
accidental: they are the fixed points of local improvement on the
strategy lattice, and they coincide with the constrained optima that
equalize marginal information per cost. And reliability has hard limits:
the information the verifier carries, the diversity of the solvers, and
the absence of a universal free lunch.

Three implications stand out. (i) \emph{Build checkers, not just
generators}: when a discriminating verifier exists, verification beats
voting by orders of magnitude (\S\ref{sec:frontier}), so the highest-value
engineering is often in the test, not the attempt. (ii) \emph{Cultivate
diversity}: cloning a single solver gains nothing beyond its shared-error
floor (Theorem~\ref{thm:diversity}). (iii) \emph{Expect emergent
layering}: any system that locally trades cost for reliability will grow
verification layers on its own (Theorem~\ref{thm:fixed}), so the design
question is less whether to layer than where the marginal information is.

\paragraph{Why organizations look the way they do.}
The same laws explain recurring features of human and biological
problem-solving. Peer review, replication, and reproducibility are
verification gates; the scientific premium on \emph{independent}
confirmation is Corollary~\ref{cor:neff} in cultural form. Bureaucratic
sign-off chains are gate cascades trading latency for reliability.
Modular redundancy in engineering (running three units and taking the
majority) is voting above threshold. Division of labor is sequential
decomposition made safe by local checking (inspection, unit tests).
That such different systems converge on layered, checker-saturated
structure is, on this account, not convergent accident but the shared
fixed point of local improvement under the reliability laws
(Theorem~\ref{thm:fixed}).

\paragraph{Limitations and future work.}
Our assumptions of independence and conditional independence are
idealized;
\S\ref{sec:lower} bounds the cost of violating them but a tight theory of
\emph{partially} correlated verifier cascades remains open. We treated
verifier quality as given; \emph{learning} the decomposition and the verifier, together with
the game-theoretic dynamics when generators adapt to checkers, is the
natural sequel, as is a quantitative bridge from these
laws to the behavior of real model harnesses, which the companion paper
takes up empirically.

\end{document}